\documentclass{new_tlp}
\pdfoutput=1
\usepackage[utf8]{inputenc}
\usepackage[T1]{fontenc} 
\usepackage{graphicx}
\usepackage{wrapfig}
\usepackage{url}
\usepackage{textcomp}
\usepackage{listings}
\usepackage[scaled=0.85]{couriers}
\usepackage{float}
\usepackage{verbatim}

\newfloat{algorithm}{tbh}{loa}
\floatname{algorithm}{Alg.}
\newcounter{AlgorithmOne}[algorithm]

\newtheorem{proposition}{Proposition}

\title[Tabling, Rational Terms, and Coinduction Finally Together!]
      {Tabling, Rational Terms, and Coinduction\\Finally Together!}

\author[T. Mantadelis, R. Rocha and P. Moura]
       {THEOFRASTOS MANTADELIS, RICARDO ROCHA and PAULO MOURA\\
       CRACS \& INESC TEC, Faculty of Sciences, University of Porto\\
       Rua do Campo Alegre, 1021/1055, 4169-007 Porto, Portugal\\
       \email{\{theo.mantadelis,ricroc\}@dcc.fc.up.pt,pmoura@inescporto.pt}
      }

\begin{document}

\label{firstpage}

\maketitle


\begin{abstract}
Tabling is a commonly used technique in logic programming for avoiding cyclic behavior of logic programs and enabling more declarative program definitions. Furthermore, tabling often improves computational performance. Rational term are terms with one or more infinite sub-terms but with a finite representation. Rational terms can be generated in Prolog by omitting the occurs check when unifying two terms. Applications of rational terms include definite clause grammars, constraint handling systems, and coinduction.
In this paper, we report our extension of YAP's Prolog tabling mechanism to support rational terms. We describe the internal representation of rational terms within the table space and prove its correctness. We then use this extension to implement a tabling based approach to coinduction. We compare our approach with current coinductive transformations and describe the implementation. In addition, we present an algorithm that ensures a canonical representation for rational terms.
\end{abstract}



\section{Introduction}

Tabling~\cite{Chen-96} is a recognized and powerful implementation
technique that solves some limitations of Prolog's operational
semantics in dealing with left-recursion and redundant
sub-computations. Tabling based models are able to reduce the search
space, avoid looping, and always terminate for programs satisfying the
\emph{bounded term-size property}\footnote{A logic program has the
  bounded term-size property if there is a function $f:N \rightarrow
  N$ such that whenever a query goal \emph{Q} has no argument whose
  term size exceeds \emph{n}, then no term in the derivation of
  \emph{Q} has size greater than $f(n)$.}. Tabling consists of saving
and reusing the results of sub-computations during the execution of a
program. This is accomplished by storing the calls and the answers to tabled subgoals in a proper data structure called the \emph{table
  space}.

From as early as~\cite{Colmerauer-1982,Jaffar-1986}, Prolog implementers have chosen to omit the \emph{occurs check} in unification. This has resulted in generating cyclic terms known as \emph{rational terms} or \emph{rational trees} in Prolog. While the introduction of cyclic terms in Prolog was unintentional, soon after applications for cyclic terms emerged in fields such as definite clause grammars~\cite{Colmerauer-1982,Giannesini-1984}, constraint programming~\cite{Meister-2006,Bagnara-2001} and coinduction~\cite{Gupta-2007}.
But support for rational terms across Prolog systems varies and often fails to provide the functionality required by most applications.
Two common problems are the lack of support for printing query bindings with rational terms and the lack of a tabling mechanism that supports rational terms~\cite{Moura-2013}. Furthermore, several Prolog features are not designed for compatibility with rational terms and can make programming using rational terms challenging and cumbersome. 

In this paper, we first present an extension to the tabling mechanism of YAP Prolog~\cite{CostaVS-12} to support rational terms. To the best of our knowledge, this is the first Prolog built-in tabling system that supports rational terms. The tabling system of XSB~\cite{Swift-12} handles infinite terms by defining a limit on the term size stored within the table space. While this approach allows tabling to work with goals containing rational terms, it does not store a rational term within the table space and, as a result, the returned answers are not the expected ones. Our approach extents the table space to represent rational terms internally. Consequently, answers containing rational terms are properly stored and returned by the tabling mechanism.

Furthermore, by having extended tabling to support rational terms, we present a novel approach for calculating the \emph{greatest fixed point} of goals by using the internal tabling mechanism. In this way, we propose a novel and efficient transformation to implement coinduction. Previous research presented different solutions to support coinduction~\cite{Lars-2007,Gupta-2007,Ancona-2013,Moura-2013} but, to the best of our knowledge, our approach is the first one that uses a native transformation, allowing a more elegant and efficient implementation. In addition, we present an algorithm that ensures a canonical representation for rational terms. 

The remainder of the paper is organized as follows. First in Section~\ref{sec:background}, we introduce some background concepts about tabling, rational terms and coinduction. Next, Section~\ref{sec:tabling_rational_trees} describes our extension to support rational terms within the tries of the table space. Then, we introduce our proposal to implement coinduction through the internal tabling mechanism in Section~\ref{sec:tabling_coinduction}. Section~\ref{sec:canonical_term} addresses the issue of rational terms in canonical form and finally, Section~\ref{sec:conclusions}, concludes by discussing related and future work.


\section{Background}
\label{sec:background}

In this section, we present some background information on tabling, rational
terms, and coinduction.


\subsection{Tabling and Tries}
\label{sec:tabling_tries}

The basic idea behind tabling is straightforward: programs are
evaluated by storing answers for tabled subgoals in an appropriate
data space called the \emph{table space}. Repeated calls to tabled
subgoals are resolved whenever possible by consuming the answers already stored in the table space instead of being re-evaluated against the program clauses.
During this process, as new answers are found, they are stored in their tables and later returned to all repeated calls.

The design of the table space data structures is critical for achieving an efficient tabling implementation. In YAP, as in most tabling systems, the table space representation is based on \emph{tries}~\cite{RamakrishnanIV-99}. A trie is a tree structure where each different path through the tree nodes corresponds to
a term described by the tokens labeling the nodes traversed. For
example, the tokenized form of the term \lstinline{path(X,1,f(Y))} is
the sequence of 5 tokens \lstinline{path/3}, \lstinline{VAR0},
\lstinline{1}, \lstinline{f/1} and \lstinline{VAR1}, where each
variable is represented as a distinct \lstinline{VARi} constant. Two
terms with common prefixes will branch off from each other at the
first distinguishing token. Consider, for example, a second term
\lstinline{path(Z,1,b)}. Since the main functor and the first two arguments, tokens
\lstinline{path/3}, \lstinline{VAR0} and \lstinline{1}, are common to both terms, only
one additional node will be required to fully represent this second
term in the trie, thus saving the space that would be taken by three nodes in this case.

YAP implements tables using two levels of tries~\cite{RamakrishnanIV-99}. The first level,
the \emph{subgoal trie}, stores the tabled subgoal calls. The
second level, the \emph{answer trie}, stores the answers for a given
call. More specifically, each tabled predicate has a \emph{table
  entry} data structure assigned to it, acting as the entry point for
the predicate's subgoal trie. Each different subgoal call is then
represented as a unique path in the subgoal trie, starting at the
table entry and ending in a \emph{subgoal frame} data structure, with
the argument terms being stored within the path's nodes. The subgoal
frame data structure acts as an entry point to the answer
trie. Contrary to subgoal tries, answer trie paths hold just the
substitution terms for the free variables that exist in the argument
terms of the corresponding call.

\begin{wrapfigure}{r}{4.5cm}
\vspace{-\intextsep}
\centering
\includegraphics[width=4.5cm]{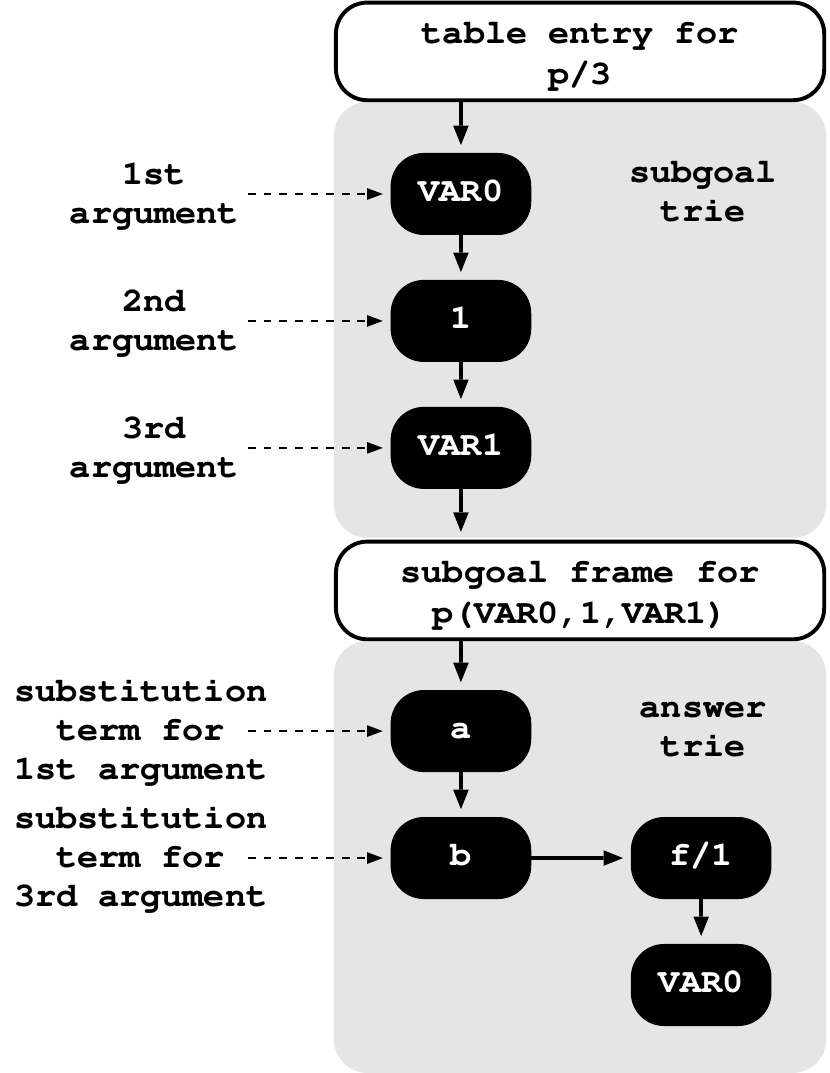}
\caption{YAP's table space organization}
\label{fig_table_design}
\vspace{-\intextsep}
\end{wrapfigure}

An example for a tabled predicate \lstinline{p/3} is shown in
Fig.~\ref{fig_table_design}. Initially, the table entry for
\lstinline{p/3} points to an empty subgoal trie. When the subgoal
\lstinline{p(X,1,Y)} is called, three trie nodes are inserted to
represent the arguments in the call: one for the variable \lstinline{X}
(\lstinline{VAR0}), a second for the integer \lstinline{1}, and third one for
the variable \lstinline{Y} (\lstinline{VAR1}). Since the predicate's
functor term is already represented by its table entry, we can avoid
inserting an explicit node for \lstinline{p/3} in the subgoal
trie. Finally, the leaf node is set to point to a subgoal frame, from
where the answers for the call will be stored. The example shows two
answers for \lstinline{p(X,1,Y)}: \lstinline|{X=a, Y=f(VAR0)}| and
\lstinline|{X=a, Y=b}|. Since both answers have the same substitution
term for argument \lstinline{X}, they share the top node in the answer
trie. For argument \lstinline{Y}, each answer has a different
substitution term and, thus, a different path is used to represent
each.


\subsection{Rational Terms}
\label{sec:rational_terms}

Rational terms, also known as rational trees or cyclic terms, are
infinite terms that can be finitely represented. They can include any
finite sub-term but have at least one infinite sub-term. Rational
terms in logic programming appeared as a side effect of omitting the
\emph{occurs check} in unification. By omitting the occurs
  check, unification can return cyclic terms. A simple example is 
\lstinline{L=[1|L]}, where the variable \lstinline{L} is instantiated to an infinite list of ones. 
Prolog implementers omitted the occurs check in order to reduce the unification complexity from
$O(Size_{Term1}+Size_{Term2})$ to $O(min(Size_{Term1},Size_{Term2}))$,
but since then these cyclic terms found application in areas such as
coinduction, natural language processing or finite-state automata, among others.

Several Prolog systems support the creation and unification of rational terms. These systems also provide built-in predicate implementations capable of handling rational terms without falling into infinite computations. Still, several common predicates such as \lstinline{member/2} cannot handle rational terms.

\begin{lstlisting}
member(E, [E|_]).
member(E, [_|T]) :- member(E, T).
\end{lstlisting}

Consider the call \lstinline{L=[1,2|L], member(E, L)}. The second
clause of \lstinline{member/2} spawns a new member call for every next
suffix of the infinite list \lstinline{L}, generating an infinite
number of calls and answers for \lstinline{member/2}. In order to
handle the occurring cycles, one could implement the following
\lstinline{member/2} predicate.
\begin{lstlisting}
member(E, L) :- member(E, L, []).

member(E, [E|_], _).
member(_, L, S) :- in_stack(L, S), !, fail.
member(E, [H|T], S) :- member(E, T, [[H|T]|S]).
\end{lstlisting}
With \lstinline{in_stack/2} predicate being defined as:
\begin{lstlisting}
in_stack(E, [H|_]) :- E == H.
in_stack(E, [_|T]) :- in_stack(E, T).
\end{lstlisting}

The second clause of \lstinline{member/3} checks whether the input list
of the member call is repeated. If so, it terminates the execution by
removing any left choice points and then fails. Using our example list
\lstinline{L=[1,2|L]}, it will permit only the two unique suffixes
of the list \lstinline{L}.
Unfortunately, this is more of an ad-hoc solution instead of a clean scalable solution that deals with generic rational terms.


\subsection{Coinduction}

Recently, we have seen an increase of interest in coinduction. Several new systematic approaches have appeared and important applications have been noted. co-SLD extends logic programming to allow reasoning over infinite and cyclic structures. For this paper, we follow similar semantics like the ones proposed by~\cite{Gupta-2007} for co-SLD: the \textit{coinductive hypothesis rule} states that \emph{if while proving a goal G (ancestor goal) a subgoal G' (current goal) is found, such as goal G is a variant of G', then the subgoal G' has coinductive success}.

We have seen several approaches to implement co-SLD with the use of program transformation~\cite{Lars-2007,Gupta-2007,Moura-2013} or by using a meta caller approach~\cite{Ancona-2013}. While SLD resolution proves a query and returns the answers of the least fixed points, co-SLD resolution succeeds and returns the answers of the greatest fixed points~\cite{Lars-2007}. As an example, consider the following logic program and the query \lstinline{bin(X)}.
\newpage
\begin{lstlisting}
bin([]).
bin([0|T]) :- bin(T).
bin([1|T]) :- bin(T).
\end{lstlisting}

SLD resolution recognizes all finite lists containing only
\lstinline{0}s and \lstinline{1}s. Now, if we remove the base case, the program does not have
a least fixed point and, as a consequence, SLD resolution will not be
able to solve it. As pointed out by~\cite{Gupta-2007}, there are two
reasons why the program is not solved. First, because the Herbrand
universe does not allow infinite terms and, second, because the
Herbrand universe does not allow infinite proofs. co-SLD allows infinite terms and infinite proofs, as a result the query \lstinline{bin(X)} succeeds and resolution returns two answers that represent the minimal greatest fixed points, \lstinline+{X=[0|X], X=[1|X]}+. However, any infinite pattern of \lstinline{0}s and \lstinline{1}s that can be expressed as a rational term is recognized by the \lstinline{bin/1} coinductive predicate, as exemplified below.
\begin{lstlisting}
% generating solutions for the coinductive bin/1 predicate:
Query:   ?- bin(X).
Answers: X = [0|X] ? ;
         X = [1|X].

% recognizing solutions for the coinductive bin/1 predicate:
Query:   ?- X = [0,1,0,1,0,0,0|X], bin(X).
Answer:  X = [0,1,0,1,0,0,0|X].
\end{lstlisting}


\section{Extending Tabling to Support Rational Terms}
\label{sec:tabling_rational_trees}

To the best of our knowledge, no current tabling system
supports rational terms. A tabled subgoal call or tabled
answer containing rational terms would result in an infinite
computation. The reason for this behavior is that the table space data
structures were not designed to support the finite storing of infinite
terms. YAP is thus the first tabling system that handles rational
terms by extending the table space to support the internal
representation of infinite terms.

Consider, for example, the definition given in~\cite{Gupta-2007} for a
\lstinline{comember/2} coinductive predicate:
\begin{lstlisting}
% comember(Element, List) is true iff Element appears infinite times in List
:- coinductive(comember/2).
comember(H, L) :- drop(H, L, L1), comember(H, L1).

% drop(Element, List, Rest) is true iff Element can be dropped from List resulting 
% in the list Rest 
:- table(drop/3).
drop(H, [H|T], T).
drop(H, [_|T], T1) :- drop(H, T, T1).
\end{lstlisting}

Without tabling with support for rational terms, the auxiliary predicate \lstinline{drop/3} can not be used by the \lstinline{comember/2} predicate to enumerate all the distinct elements of a (rational) list. Only the first element would be enumerated. To avoid tabling, one could write \lstinline{drop/3} as follows, by using the \lstinline{in_stack/2} auxiliary predicate, but certainly writing such a procedural predicate is not in the spirit of logic programming. Similarly, the solution presented by~\cite{Ancona-2013} that avoids tabling by using coinductive success hook predicates requires from the programmer a good understanding of the coinductive meta-interpreter or coinductive transformation mechanics. Ancona's coinductive success hook predicates are not illustrated here for brevity. The crucial concept behind them, is to enable the programmer to modify the behavior when a cycle is detected.
\begin{lstlisting}
drop(E, L, NL) :- drop(E, L, NL, []).
drop(_, L, _, S) :- in_stack(L, S), !, fail.
drop(E, [E|T], T, _).
drop(E, [H|T], T1, S) :- drop(E, T, T1, [[H|T]|S]).
\end{lstlisting}

Motivated by the above example and the goal to elegantly and efficiently support coinduction for YAP at
the low level engine, our first step was to extend YAP's internal tries
representation to support rational terms. So far, tries do not
support cyclic terms and inserting one such term in a trie results in
trapping the program in a cycle. In order for tries to support cyclic
terms, we first defined a representation of cyclic tokens within the
tries; we then extended the trie check/insert algorithm to
detect the tokens where a cycle occurs and at which upper trie token
it should point to; last, we adjusted the mechanism that allows to
reconstruct a term from the stored tokens within the trie to also
handle cyclic terms.

\begin{wrapfigure}{r}{6.5cm}
\vspace{-\intextsep}
\includegraphics[width=6.5cm]{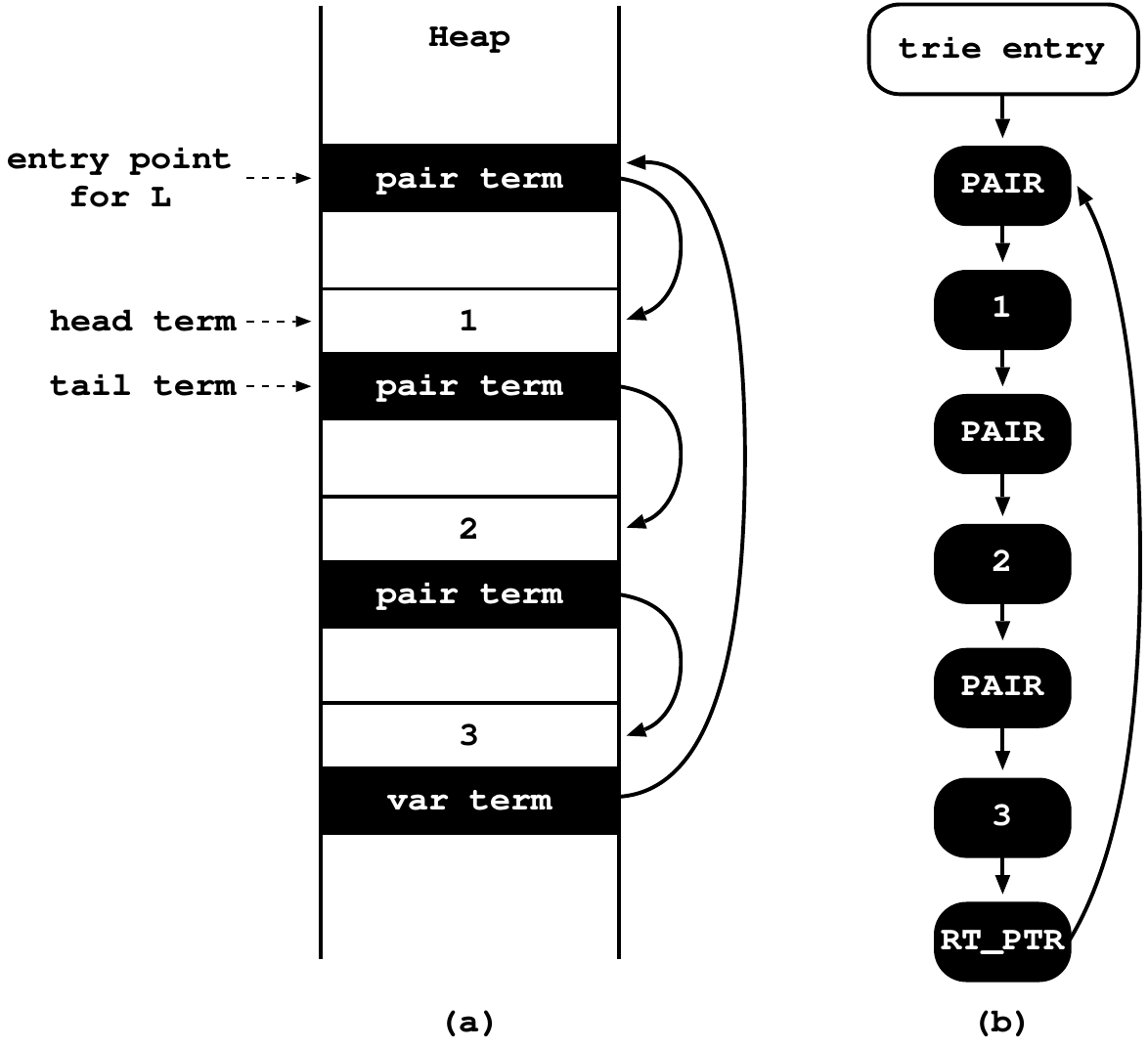}
\caption{YAP's representation for the cyclic term \lstinline{L=[1,2,3|L]} in the \textbf{(a)} heap and \textbf{(b)} the trie}
\label{fig:cyclic_term}
\vspace{-\intextsep}
\end{wrapfigure}

To better explain our trie representation for cyclic terms let's have a look at the heap representation. At the low level engine, a cycle in YAP appears as a \emph{variable term} cell unified with the term that is repeated. Figure~\ref{fig:cyclic_term}(a) shows YAP's internal representation for the cyclic term \lstinline{L=[1,2,3|L]} as it is stored in the heap.

Current trie definitions do not store pointer tokens by default. We extend the trie definition so that it can store a pointer token that points to a trie node. As no other trie nodes are storing pointers, we can safely assume that such a pointer token represents a rational term cycle. In order to represent the cycle within the trie, we store in the node of the cyclic token a pointer to the trie node of the pointed token. Figure~\ref{fig:cyclic_term}(b) illustrates the cyclic term \lstinline{L=[1,2,3|L]} as it is stored in the trie representation\footnote{For the sake of simplicity, we describe YAP's original trie implementation for lists. In fact, our current implementation extends the optimized \emph{compact lists} representation for tries as described in~\cite{Raimundo-10}.} with \lstinline{RT_PTR} illustrating the rational term pointer token. In order for our approach to be safe we consider Proposition~\ref{proposition:cyclic_term}.
\begin{proposition}
Cycles in a Prolog term can only point at a list (pair) term or a functor term.
\label{proposition:cyclic_term}
\end{proposition}
\begin{proof}
If a variable term cell points to an atomic term or to an unbound
variable then it would be dereferenced to that atomic term or unbound
variable.
\end{proof}

We detect the cycles upon the insertion of a term in the trie. To do
so, we first memorize the heap address of every list or functor term
that we encounter \textsl{tuppled} with the address of the trie node that
represents the list or functor token in the trie. Before inserting a
new list or functor token in the trie, we check whether we have
already inserted a node for it. If so, we return the address of the
corresponding trie node and create a new trie node containing a pointer token with the address of the returned trie node.

Besides storing a cyclic term, the tabling system must also be able to return the same cyclic term when asked. In YAP Prolog, to reconstruct a term from the trie, the tabling mechanism reads the trie branch bottom up and constructs the subterms stored in each token incrementally directly in the Prolog heap. The original term reconstruction mechanism assumes that no reference pointers are stored within the trie nodes. As the construction of the term is bottom up our system does not know the reference of a cycle before fully constructing the term. In order to tackle this issue when a reference pointer is found, we create a new unbound variable in the heap and memorize the heap address and the trie node that has the token which will be the entry point of the term that this variable must be unified. When the term has been fully constructed, we unify that unbound variable with the term generated from the corresponding trie node that the pointer token is pointing at. In order for our approach to be safe we considered Proposition~\ref{proposition:cyclic_earlier}.
\begin{proposition}
When traversing a Prolog term from the heap entry point, cyclic references point at preceding heap locations.
\label{proposition:cyclic_earlier}
\end{proposition}
\begin{proof}
Indeed, if while traversing a Prolog term, starting from the entry point, a variable term cell points to a following heap location it does not introduce a cycle.
\end{proof}

Returning to our \lstinline{member/2} example from
Section~\ref{sec:rational_terms}, we now present an alternative more
elegant solution to handle lists with rational terms.

\begin{lstlisting}
:- table(member/2).
member(E, [E|_]).
member(E, [_|T]) :- member(E, T).
\end{lstlisting}

Please note that tabling not only remembers answers but also
handles introduced cycles. Similarly, any predicate that needs to
process a list which contains rational terms can use tabling to avoid
using a similar mechanism like the one presented at
Section~\ref{sec:rational_terms}.


\section{A Tabling Based Approach to Coinduction}
\label{sec:tabling_coinduction}

In this section, we describe a modified tabling strategy that instead
of computing the least fixed point of a logic program, computes the
greatest fixed point. In order to compute the greatest fixed point of
a logic program, it is customary to either define a program
transformation~\cite{Lars-2007,Gupta-2007,Moura-2013} or implement a
meta-interpreter~\cite{Ancona-2013} that solves the logic
program. Both approaches have been presented in the past and, to the
best of our knowledge, there is no previous implementation that computes
the greatest fix point within the WAM.

In order to better describe our tabling based approach, we first present in Alg.~\ref{alg:coinduction_program_transformation} the program transformation described in~\cite{Moura-2013}. There are three key features in the transformation:
(i) it uses a stack to keep track the ancestor calls (also called coinductive hypothesis) of the goal that the greatest fix point is under computation;
(ii) the stored ancestor calls in the stack are then used for detecting cycles (for coinduction purposes a cycle occurs when an ancestor call unifies with the current call) in the execution of the goal; and
(iii) when a cycle is detected, it is handled as a success and unifies the called goal with the repeated ancestor goal.
\begin{algorithm}
\textbf{Input:} a coinductive logic predicate \lstinline{p/1}
\begin{lstlisting}
:- coinductive(p/1).
p(Args) :-
	% the body/2 predicate abstracts the processing of the predicate arguments
	body(Args, NewArgs),
	p(NewArgs).
\end{lstlisting}
\textbf{Output:} \lstinline{p/1} transformed to call an auxiliary predicate that implements coinduction
\begin{lstlisting}
p(Args) :-
	% start with an empty stack of coinductive hypothesis
	p(Args, []).
p(Args, Stack) :-
	body(Args, NewArgs),
	p_coinduction_perflight(NewArgs, Stack).
p_coinduction_perflight(Args, Stack) :-
	(	member(p(Args), Stack) *->
		% coinductive success;
		% the (*->)/2 soft-cut construct supports finding of multiple solutions
		true
	;	% add current goal to the stack of coinductive hypothesis and continue
		p(Args, [p(Args)|Stack])
	).
\end{lstlisting}
\caption{The coinductive program transformation of~\protect\cite{Moura-2013}.}
\label{alg:coinduction_program_transformation}
\end{algorithm}

Algorithm~\ref{alg:coinduction_program_transformation}
is akin both to the cycle handling used in the PTTP~\cite{Stickel-1988} and in tabling
systems. Tabling memoizes subgoals in a table and checks whether there
is a repeated call in order to break the cycle, thus allowing cyclic
programs to terminate. Conceptually, the subgoal memoization mechanism
of tabling used to perform cycle detection is similar to the
collection of calls in the stack used by
Alg.~\ref{alg:coinduction_program_transformation}. One slight difference is that Alg.~\ref{alg:coinduction_program_transformation} is using unification to detect repeating goals and YAP's tabling mechanism uses \emph{variant checking}\footnote{Two goals are said to be variants if they are the same up to variable renaming.}, which may alter the depth of derivations needed to detect the coinductive cycle. To compute the greatest fixed point of a logic program by reusing the mechanism of tabling, we thus noticed that we need to do two main
modifications to YAP's tabling mechanism. First, we need to modify the
behavior of the cycle detection mechanism in order to succeed the subgoal when a cycle was detected and, second, we need to unify the arguments of the original
call of the goal with the arguments of the repeated subgoal where the
cycle was detected. While the first modification is easy to motivate,
the second modification is not so obvious. To clarify the function and necessity of the second modification, we use the example shown in Fig.~\ref{fig:bin_example}.

\begin{figure}[ht]
\centering
\begin{minipage}[]{.44\linewidth}
\begin{lstlisting}
:- coinductive(bin/1).
bin([0|T]) :- bin(T).
bin([1|T]) :- bin(T).

Query:   ?- bin(L).
Answers: L = [0|L] ? ;
         L = [1|L].
Query:   ?- X=[0,1,0,0|X], bin(X).
Answer:  X = [0,1,0,0|X].
\end{lstlisting}
\end{minipage}
\begin{minipage}[c]{.45\linewidth}
\includegraphics[width=5.6cm]{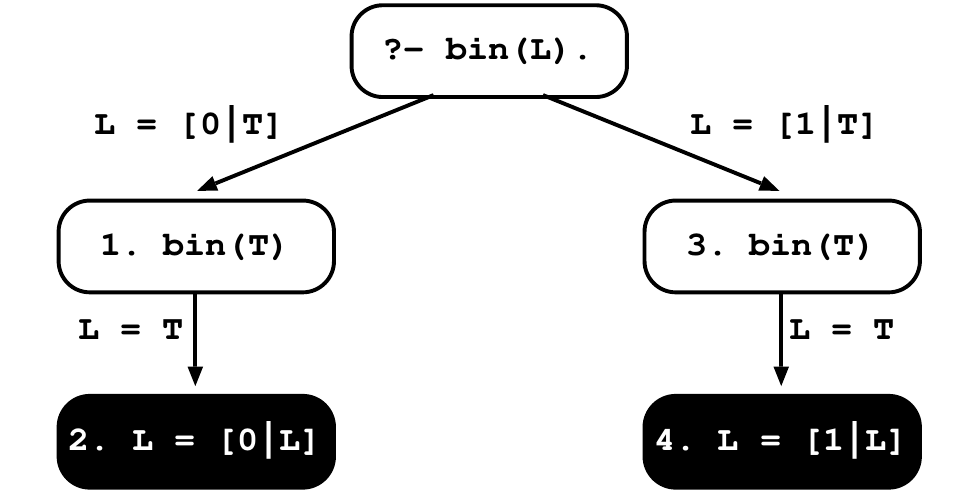}
\end{minipage}
\caption{The coinductive predicate \lstinline{bin/1} and its coinductive SLG tree}
\label{fig:bin_example}
\end{figure}

In this example, the unification of the query goal
\lstinline{bin(L)} with the head of both clauses for \lstinline{bin/1}
leaves the tail \lstinline{T} of the list unbound (\lstinline{L=[0|T]}
and \lstinline{L=[1|T]} for the first and second clause,
respectively). The cycle is then detected when the subgoal
\lstinline{bin(T)} is called next (steps 1 and 3 in
Fig.~\ref{fig:bin_example}). Unifying the arguments for the initial
query goal \lstinline{bin(L)} with the subgoal \lstinline{bin(T)}
leads to unifying \lstinline{L} with \lstinline{T}, which results in
the answers \lstinline{L=[0|L]} and \lstinline{L=[1|L]} for the query
goal (steps 2 and 4 in Fig.~\ref{fig:bin_example}).

\begin{wrapfigure}{r}{7.2cm}
\vspace{-\intextsep}
\includegraphics[width=7.2cm]{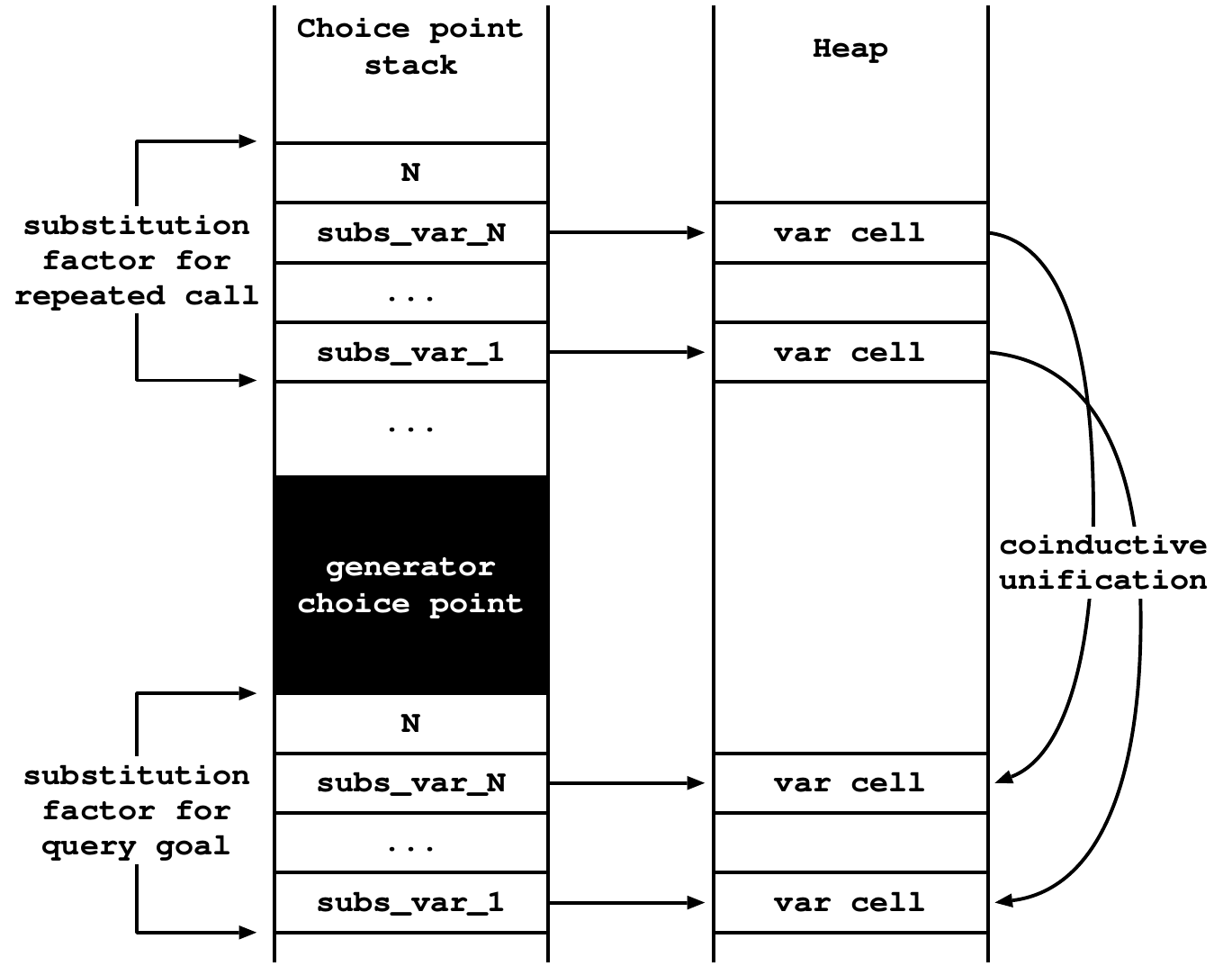}
\caption{Using the substitution factors to implement coinductive unification}
\label{fig:subs_factor}
\vspace{-\intextsep}
\end{wrapfigure}

Now that we have identified the necessary behavior for computing the greatest fixed point with co-SLG, we present how it is
actually implemented.  As explained in section~\ref{sec:tabling_tries}, the tabling mechanism retains a state called subgoal frame for each subgoal
call that is proving. In particular, this includes the
\emph{substitution factor}~\cite{RamakrishnanIV-99} for the subgoal at
hand, i.e., the set of free variables which exist within the term
arguments of the subgoal call. For example, the substitution factor
for the subgoal call \lstinline{p(X,1,f(Y))} is the triplet
\lstinline{<2,X,Y>}, where the first entry indicates the number of
variables. The substitution factor is stored on the choice
point stack during the process of checking/inserting the call in the subgoal
trie for the predicate at hand. The substitution factor is then used
to reduce the number of (substitution) terms to copy into and out of
the answer tries.

During evaluation, when a tabled subgoal is first called, it is
marked as a \emph{generator call} and tagged as \emph{evaluating}. At
the low level engine, that corresponds to storing a generator choice
point on top of the substitution factor for the subgoal. Later, if a
subgoal tagged as evaluating is encountered again then tabling treats
it as a cycle. Exactly at that part of the tabling mechanism we intervene.
Please note that at that point, we already have the substitution factor
for the repeated call also on the choice point stack. Our task is then to
unify both substitution factors in order to perform the needed
unification for the coinductive strategy. Figure~\ref{fig:subs_factor}
illustrates how this mechanism works and how we use the substitution
factors to implement coinductive unification. It presents the choice
point and heap stacks at the moment a cycle is found and after
performing coinductive unification. This unification will always
succeed as YAP uses variant checking in order to identify cycles. The next step is to follow the substitution terms for the
generator choice point and insert the resulting unification as an
answer in the table space for the generator call. Finally, the
execution succeeds and returns the answer.


\section{Canonical Representation for Rational Terms}
\label{sec:canonical_term}

In this section, we address the known issue of rational terms and their canonical or minimal representation~\cite{Colmerauer-1982}. For example, the rational terms \lstinline{A=[1|A]}, \lstinline{B=[1,1|B]}, \lstinline{C=[1|A]} are all both unifiable and equal:
\begin{lstlisting}
?- A=[1|A], B=[1,1|B], C=[1|A], A=B, B=C, A==B, B==C.
% SWI Prolog Result:
A = B, B = C, C = [1|_S1], % where
    _S1 = [1|_S1].
% YAP Prolog Result:
Stable version:      A = B = C = [1|**].
Development version: A=[1|A], B=[1,1|B], C=[1|A].
\end{lstlisting}

\begin{wrapfigure}{r}{4.5cm}
\centering
\includegraphics[width=4.5cm]{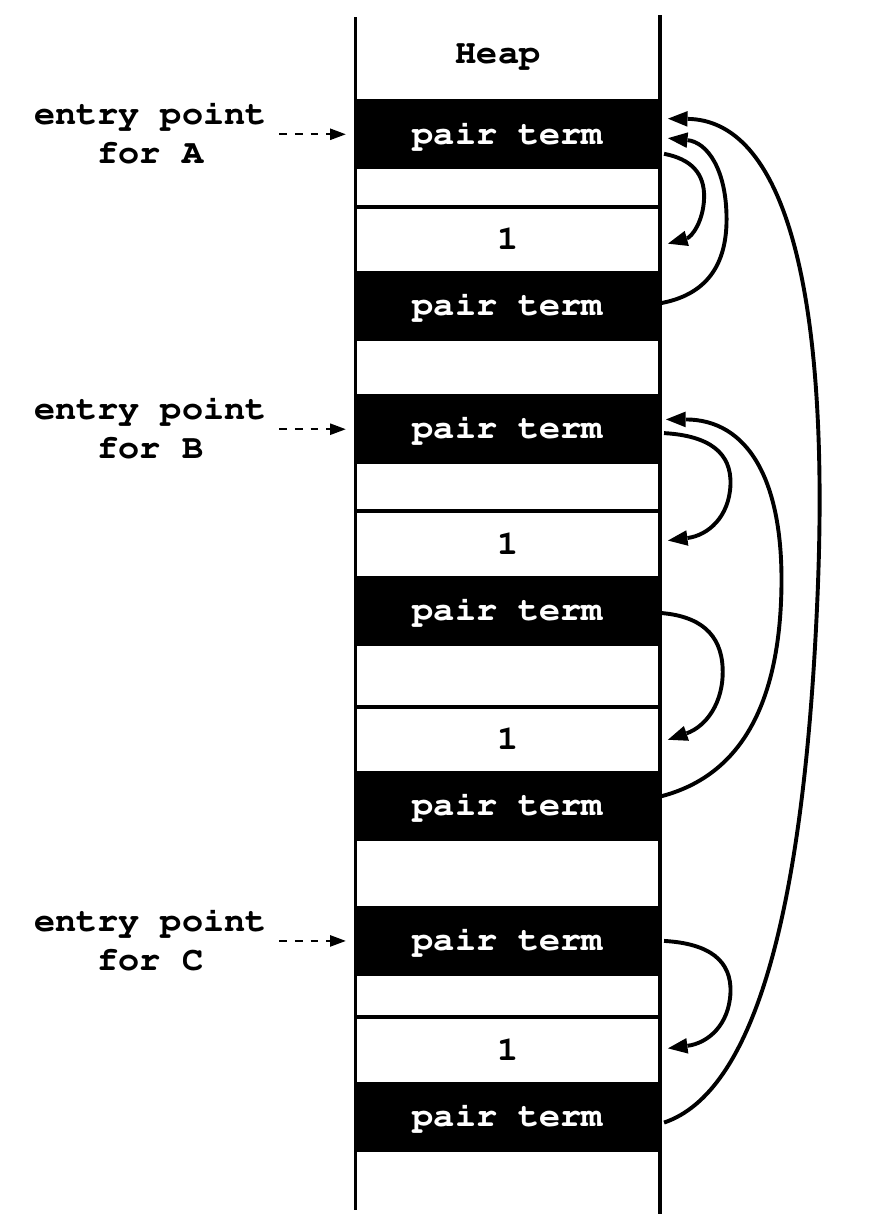}
\caption{Heap representation of rational terms \lstinline{A=[1|A]}, \lstinline{B=[1,1|B]}, \lstinline{C=[1|A]}}
\label{fig:heap}
\vspace{-\intextsep}
\end{wrapfigure}

While all three lists represent the same rational term, their underlying WAM representation in the heap may not be the same\footnote{We have verified that the terms have different heap representation in SWI and YAP Prolog systems.}. Figure~\ref{fig:heap} presents YAP's representation of each of the lists in the heap. Notice that all three lists have a different representation. This particularity of rational terms, which does not appear in normal terms, imposes a limitation for tabling. All these representations appear as different terms for tabling and create different subgoals and answers within the table space. As a result a tabled predicate may table more subgoals than needed or succeed multiple times for an answer that should succeed only once.

In order to address the problem of the extra answers, a possible solution would be to perform a kind of answer subsumption. But a more elegant solution would be to implement an algorithm that converts a rational term to its canonical form. In this section, we present such an algorithm. Predicate \lstinline{canonical_term/2} presented in Alg.~\ref{alg:canonical_term} takes any term and converts its rational parts to their canonical form.

The idea behind the algorithm is to first fragment the term to its cyclic subterms, continue by reconstructing each cyclic subterm (now acyclic) and, finally, reintroduce the cycle to the reconstructed subterms. To reconstruct each cyclic subterm as acyclic, the algorithm copies the unique parts of the term and introduces an unbound variable instead of the cyclic references. Then, the algorithm binds the unbound variable to the reconstructed subterm, recreating the cycle. Take for example the rational term \lstinline{L=[1,2,1,2|L]}, term \lstinline{L} is being fragmented in the following subterms \lstinline{L0=[1|L1]}, \lstinline{L1=[2|L3]}, \lstinline{L3=[1,2|L0]}. We do not need to fragment the term \lstinline{L3} as, at that point, our algorithm detects a cycle and replaces term \lstinline{L3} with an unbound variable \lstinline{OpenEnd}. Thus we get the following subterms \lstinline{L0=[1|L1]} and \lstinline{L1=[2|OpenEnd]}. Binding \lstinline{OpenEnd=L0} results to the canonical rational term \lstinline{L0=[1,2|L0]}.

\begin{algorithm}
\textbf{Input:} a rational term \lstinline{Term}\\
\textbf{Output:} a rational term \lstinline{Canonical} in canonical representation
\begin{lstlisting}[numbers=left]
canonical_term(Term, Canonical) :-
	Term =.. InList,
	decompose_cyclic_term(Term, InList, OutList, OpenEnd, [Term]),
	Canonical =.. OutList,
	Canonical = OpenEnd.
 
decompose_cyclic_term(_CyclicTerm, [], [], _OpenEnd, _Stack).
decompose_cyclic_term(CyclicTerm, [Term|Tail], [Term|NewTail], OpenEnd, Stack) :-
	acyclic_term(Term), !,
	decompose_cyclic_term(CyclicTerm, Tail, NewTail, OpenEnd, Stack).
decompose_cyclic_term(CyclicTerm,[Term|Tail],[OpenEnd|NewTail],OpenEnd,Stack) :-
	CyclicTerm == Term, !,
	decompose_cyclic_term(CyclicTerm, Tail, NewTail, OpenEnd, Stack).
decompose_cyclic_term(CyclicTerm,[Term|Tail],[Canonical|NewTail],OpenEnd,Stack) :-
	\+ in_stack(Term, Stack), !,
	Term =.. InList,
	decompose_cyclic_term(Term, InList, OutList, OpenEnd2, [Term|Stack]),
	Canonical =.. OutList,
	(	Canonical = OpenEnd2,
		Canonical == Term,
		!
	;	OpenEnd2 = OpenEnd
	),
	decompose_cyclic_term(CyclicTerm, Tail, NewTail, OpenEnd, Stack).
decompose_cyclic_term(CyclicTerm,[_Term|Tail],[OpenEnd|NewTail],OpenEnd,Stack) :-
	decompose_cyclic_term(CyclicTerm, Tail, NewTail, OpenEnd, Stack).
\end{lstlisting}
\caption{Predicate \lstinline{canonical_term/2}}
\label{alg:canonical_term}
\end{algorithm}

The bulk of the algorithm is found at the fourth clause of \lstinline{decompose_cyclic_term/5}. At that part we have detected a cyclic subterm that we have to treat recursively. In particular, lines 19--23 implement an important step. Returning to our example when the cycle is detected, the algorithm returns the unbound variable to each fragmented subterm. First, the subterm \lstinline{L1=[2|OpenEnd]} appears and the algorithm needs to resolve whether it must unify \lstinline{OpenEnd} with \lstinline{L1} or whether \lstinline{OpenEnd} must be unified with a parent subterm. In order to verify that, lines 19--23 of the algorithm unify the subterm with the unbound variable and after attempt to unify the created rational term with the original rational term. For our example the algorithm generates \lstinline{L1=[2|L1]} and attempt to unify with \lstinline{L=[1,2,1,2|L]}, as the unification fails the algorithm propagates the unbound variable to be unified with the parent subterm \lstinline{L0=[1|L1]}.

The fifth clause of \lstinline{decompose_cyclic_term/5} is the location where a cycle is actually found. At that point we can drop the original cyclic subterm and place an unbound variable within the newly constructed term. The third clause of \lstinline{decompose_cyclic_term/5} could be omitted; it operates as a shortcut for simplifying rational terms of the form \lstinline{F=f(a,f(a,F,b),b)}. The rest of the algorithm is pretty much straightforward, the first and second clause of \lstinline{decompose_cyclic_term/5} are the termination condition and a copy of the non-rational parts of the term to the new term respectively.


\section{Conclusions, Related and Future Work}
\label{sec:conclusions}

In this paper, we presented a novel approach that permits the insertion of rational terms in tries and we extended tabling over goals with rational terms. Rational terms are used in several fields such as coinduction, definite clause grammars, constrain programming, lazy evaluation, finite state automata, etc. Our second novel contribution is a native-to-WAM coinductive transformation. We introduced co-SLG by presenting a tabling transformation that computes the greatest fixed point that provides an efficient elegant solution to coinduction. Finally, we provided an ISO Prolog approach that converts any rational term to its canonical form.

As this work spans over several different fields of logic programming, one finds related work from a variety of sources. In tabling, XSB Prolog~\cite{Swift-12} has limited support for infinite terms within XSB's tabling mechanism. The user can set a limit for the length of a term and define an action among \emph{error}, \emph{warning} or \emph{abstract}. When the limit is reached the system performs the chosen action. This mechanism, while it can ensure a cyclic safe program, is not suitable to handle rational terms as entities neither it is suitable to be used for coinduction purposes.

In the area of coinduction, we find similarities in the work of~\cite{Lars-2007,Gupta-2007,Moura-2013} that makes clear the importance of tabling for coinduction. \cite{Ancona-2013} presents an alternative approach to coinduction where tabling is not required, still this approach may hurt readability of the coinductive predicate definitions. Coinduction for functional languages is being a topic of research for some time now~\cite{Gordon-94} and it is been used to prove properties of lazy streams.

The problem of equivalence for finite automata~\cite{Hopcroft-71} is similar with the problem of reducing a rational term to its canonical form. Recently,~\cite{Schrijvers-2012} presented an approach to check equality of rational terms in functional programming languages. Finally, SWI Prolog~\cite{Wielemaker-12} implements a \lstinline{term_factorized/3} predicate that collects all subterms that are used multiple times in a term and substitutes them by variables. While \lstinline{term_factorized/3} serves a different purpose, the generated subterms are in canonical form and one could use this predicate to create the canonical form of a rational term.

As this is the first tabling mechanism that supports rational terms there is room for improvements and future work. On the tabling side, we should migrate our extension to work with more tabling strategies and options. On the side of rational terms, we intent to further investigate transforming rational terms to their canonical form.


\section*{Acknowledgments}
The authors want to thank Vítor Santos Costa for his suggestions and technical support. We also want to thank the anonymous reviewers for their comments and help to improve our paper. This work is partially funded by the ERDF (European Regional Development Fund) through the COMPETE Programme and by FCT (Portuguese Foundation for Science and Technology) within projects SIBILA (NORTE-07-0124-FEDER-000059) and PEst (FCOMP-01-0124-FEDER-037281).

\bibliographystyle{acmtrans}

\appendix

\section*{YAP Installation and Rational Term Support}

At the time of this submission, our contributions are part of the development version of YAP (\lstinline{git clone git://git.code.sf.net/p/YAP/YAP-6.3}). Inside the folder \lstinline{YAP-6.3/} the \lstinline{ICLP2014_examples.YAP} file contains the paper examples using the appropriate tabling settings. Currently, tabling supports rational terms only when the answers are loaded by the tries. To activate that option one can use the following YAP flag: 
\begin{lstlisting}
yap_flag(tabling_mode, load_answers). 
\end{lstlisting}
Furthermore, our coinductive transformation can be activated by the following directives:
\begin{lstlisting}
:- table PREDICATE/ARITY.
:- tabling_mode(PREDICATE/ARITY, coinductive).
\end{lstlisting}

\textbf{Technical Requirements:} Our tabling extensions only require that the Prolog system allows the creation of rational terms and that unification (\lstinline{=/2}) works between rational terms. Our \lstinline{canonical_term/2} predicate also requires that the operators \lstinline{==/2}, \lstinline{=../2} work with rational terms.

\section{Coinduction Examples}
The following examples are recreations of the examples presented at~\cite{Moura-2013} by using our implementation.
\begin{lstlisting}
:- table(comember/2).
:- tabling_mode(comember/2, coinductive).

% Returns the infinite members of a list.
comember(H, L) :-
  drop(H, L, L1),
  comember(H, L1).

:- table(drop/3).
drop(H, [H|T], T).
drop(H, [_|T], T1) :- drop(H, T, T1).

% Queries:
?- _L=[1,2|_B], _B=[3,4,5|_B], comember(E, _L).
E = 3 ? ;
E = 4 ? ;
E = 5 ? ;
false.

?- comember(1, A).
A = [1|A] ? ;
A = [_1,1,_1|A] ? ;
A = [_1,_2,1,_1|A] ? ;
A = [_1,_2,_3,1,_1|A] ? 
...

?- A=[1,2,3|A], drop(H, A, T).
A = [1,2,3|A],
H = 1,
T = [2,3,1|T] ? ;
A = [1,2,3|A],
H = 2,
T = [3,1,2|T] ? ;
A = T = [1,2,3|A],
H = 3 ? ;
false.

?- B=[1|A],A=[2,3|A], drop(H, B, T).
A = T = [2,3|A],
B = [1,2,3|B],
H = 1 ? ;
A = [2,3|A],
B = [1,2,3|B],
H = 2,
T = [3,2|T] ? ;
A = T = [2,3|A],
B = [1,2,3|B],
H = 3 ? ;
false.

?- A=[1,2,3|A], member(H, A).
A = [1,2,3|A],
H = 1 ? ;
A = [1,2,3|A],
H = 2 ? ;
A = [1,2,3|A],
H = 3 ? ;
A = [1,2,3|A],
H = 1 ? ;
A = [1,2,3|A],
H = 2 ? 
...

?- B=[1|A],A=[2,3|A], member(H, B).
A = [2,3|A],
B = [1,2,3|B],
H = 1 ? ;
A = [2,3|A],
B = [1,2,3|B],
H = 2 ? ;
A = [2,3|A],
B = [1,2,3|B],
H = 3 ? ;
A = [2,3|A],
B = [1,2,3|B],
H = 2 ? ;
A = [2,3|A],
B = [1,2,3|B],
H = 3 ? 
...
\end{lstlisting}

\begin{lstlisting}
:- table(p/1).
:- tabling_mode(p/1, coinductive).

:- table(q/1).
:- tabling_mode(q/1, coinductive).

:- table(r/1).
:- tabling_mode(r/1, coinductive).

% Tangle example.
p([a|X]) :- q(X).
p([c|X]) :- r(X).
q([b|X]) :- p(X).
r([d|X]) :- p(X).

% Queries:
?- p(X).
X = [a,b|X] ? ;
X = [c,d|X].

?- L = [a,b,c,d|L], p(L).
L = [a,b,c,d|L].

?- L = [a,c|L], p(L).
false.
\end{lstlisting}

\begin{lstlisting}
:- table(automaton/2).
:- tabling_mode(automaton/2, coinductive).

% Automaton example.
automaton(State, [Input|Inputs]) :-
  trans(State, Input, NewState),
  automaton(NewState, Inputs).

trans(s0, a, s1).
trans(s1, b, s2).
trans(s2, c, s3).
trans(s2, e, s0).
trans(s3, d, s0).

% Queries:
?- automaton(s0, X).
X = [a,b,c,d|X] ? ;
X = [a,b,e|X].

?- L = [a,b,c,d,a,b,e|L], automaton(s0, L).
L = [a,b,c,d,a,b,e|L].

?- L = [a,b,e,c,d|L], automaton(s0, L).
false.
\end{lstlisting}

\begin{lstlisting}
:- table(sieve/2).
:- tabling_mode(sieve/2, coinductive).

:- table(filter/3).
:- tabling_mode(filter/3, coinductive).

% computes a coinductive list with all the primes in the 2..N interval
primes(N, Primes) :-
  generate_infinite_list(N, List),
  sieve(List, Primes).

% generate a coinductive list with a 2..N repeating patern
generate_infinite_list(N, List) :-
  sequence(2, N, List, List).

sequence(Sup, Sup, [Sup| List], List) :-
  !.
sequence(Inf, Sup, [Inf| List], Tail) :-
  Next is Inf + 1,
  sequence(Next, Sup, List, Tail).

sieve([H| T], [H| R]) :-
  filter(H, T, F),
  sieve(F, R).

filter(H, [K| T], L) :-
  (  K > H, K mod H =:= 0 ->
     % throw away the multiple we found
     L = T1
  ;  % we must not throw away the integer used for filtering
     % as we must return a filtered coinductive list
     L = [K| T1]
  ),
  filter(H, T, T1).

% Queries:
?- primes(20, P).
P = [2,3,5,7,11,13,17,19,2,3,5,7,11,13,17,19,2,3|P].
\end{lstlisting}

\section{Experiments}
We used the following example program in order to run experiments.
\begin{lstlisting}
:- table(path/2).
:- tabling_mode(path/2, coinductive).

% Finds infinite paths starting from node F.
path(F, [F|P]) :-
  edge(F, N),
  path(N, P).

edge(1, 2).
edge(1, 3).
edge(2, 4).
edge(2, 3).
edge(3, 2).

% Queries:
?- path(1, P).
P = [1,2,3|P] ? ;
P = [1,3,2|P].

?- path(2, P).
P = [2,3|P].

?- path(3, P).
P = [3,2|P].

?- path(4, P).
false.
\end{lstlisting}

Instead of the small graph shown above, we used a fully connected graph of different sizes, as presented next:

\begin{lstlisting}
full_edge_size(8).

edge(X, Y) :-
  posint(X),
  posint(Y),
  X \== Y.

posint(N) :-
  posint(N, 0).
posint(_, I) :-
  full_edge_size(N),
  I > N, !,
  fail.
posint(I, I).
posint(X, I) :-
  NI is I + 1,
  posint(X, NI).
\end{lstlisting}

And implemented the same program in Logtalk:

\begin{lstlisting}
:- object(path).
  :- public(path/2).
  :- coinductive(path/2).
%  :- table(path/2). % used for tabled co-SLD

  path(F, [F|P]) :-
    edge(F, N),
    path(N, P).

  full_edge_size(8).

  edge(X, Y) :-
    posint(X),
    posint(Y),
    X \== Y.

  posint(N) :-
    posint(N, 0).
  posint(_, I) :-
    full_edge_size(N),
    I > N, !,
    fail.
  posint(I, I).
  posint(X, I) :-
    NI is I + 1,
    posint(X, NI).
:- end_object.
\end{lstlisting}

We executed the query: \lstinline{time((path(1, _P), fail))} on graphs with size 8x8 up to 19x19. Table~\ref{tbl:appendix} presents our results; all times are in seconds. For co-SLD we used the coinductive transformation of Logtalk. For co-SLG we used our transformation. Furthermore, we also experimented with the \lstinline{path/2} coinductive predicate in Logtalk being tabled with our rational term support. This can be achieved by just tabling the co-SLD transformed predicate of Logtalk as shown at the comment instruction of the Logtalk code.

\begin{table}[ht]
\caption{Experimental results for the query \lstinline{(path(1, _P), fail)}}
\label{tbl:appendix}
\centering
\begin{tabular}{crrr}
\textbf{Graph Size} & \textbf{co-SLD} & \textbf{co-SLG} & \textbf{Tabled co-SLD} \\\hline
8x8	&	1	&	0.005 & 4 \\
9x9	&	11	&	0.014 & *	\\
10x10	&	126	&	0.035 &	\\
11x11	&     1,724	&	0.053 &	\\
12x12	& > 324,000     &	0.126 &	\\
13x13	&		&	0.287 &	\\
14x14	&		&	0.674 &	\\
15x15	&		&	1.5 &	\\
16x16	&		&	3.5 &	\\
17x17	&		&	8 &	\\
18x18	&		&	17 &	\\
19x19	&		&	39 &	\\
20x20   &               &        * & \\
\end{tabular}
\end{table}

As expected the difference for the path example for co-SLD and co-SLG is significant. This is easy to explain as tabling in this case decreases the complexity of the problem. co-SLD is able to solve up to an 11x11 graph in approximately 1,700 seconds, and it takes more than 1.5 hours which was used as our timeout to solve the 12x12 problem. On the other hand, enumerating all the cyclic paths with co-SLG is speed efficient but very memory consuming. Our system exhausted the 10GB memory that was available to it trying to calculate the 20x20 graph. We also executed experiments by tabling the co-SLD approach. These experiments, which obtained the worst results, are important in order to point out that co-SLG is not simply achieved by tabling a co-SLD transformed predicate. The tabled co-SLD predicate soon went out of memory as it had the task to table both the resulting paths and the coinductive hypothesis. Unfortunately, we could not ignore the coinductive hypothesis from being tabled due to a limitation of YAP's mode-directed tabling~\cite{Santos-12}.

It is well known that using tabling naively can result on a poor performance; for example tabling \lstinline{append/2} predicate as shown in~\cite{Swift-12}. The same applies for using co-SLG over co-SLD naively. There are examples that co-SLD will perform better than co-SLG like the Eratosthenes sieve example.

\label{lastpage}
\end{document}